\documentclass[]{interact}

\usepackage[natbibapa,nodoi]{apacite}
\usepackage[utf8]{inputenc}
\usepackage{amsmath}
\usepackage{amsfonts}
\usepackage{amssymb}
\usepackage{amsthm}
\usepackage{graphicx}
\usepackage{pifont}
\usepackage{enumitem}
\usepackage{float}
\usepackage{color}
\usepackage{hyperref}

\usepackage{ragged2e}

\newtheorem{theorem}{Theorem}
\newtheorem{corollary}[theorem]{Corollary}
\newtheorem{definition}[theorem]{Definition}
\newtheorem{example}[theorem]{Example}
\newtheorem{lemma}[theorem]{Lemma}

\newtheorem{remark}[theorem]{Remark}

\providecommand{\numberTblEq}[1]{\refstepcounter{tblEqCounter}\label{#1}\thetag{\thetblEqCounter}}

\def\sign{\hskip2pt{\rm sign}\hskip0pt}

\def\Red#1{\textcolor{red}{#1}}
\def\Blue#1{\textcolor{blue}{#1}}

\begin{document}

\title{On finite-time and fixed-time consensus algorithms for dynamic networks switching among disconnected digraphs\thanks{Partially supported by the Czech Science Foundation through the research grant No. 17-04682S.}
\thanks{\Red{This is the accepted version of the manuscript: Gómez-Gutiérrez, D., Vázquez, C. R., \v{C}elikovský, S., Sánchez-Torres, J. D., and Ruiz-León, J. (2020). On finite-time and fixed-time consensus algorithms for dynamic networks switching among disconnected digraphs. International Journal of Control, 93(9), 2120-2134. \textbf{Please cite the publisher's version}. For the publisher's version and full citation details see: \url{https://doi.org/10.1080/00207179.2018.1543896}.} \Blue{The following links provide access, for a limited time, to a free copy of the publisher's version: \href{https://www.tandfonline.com/eprint/FSW8JJRVPHMXJ3XUUXZH/full?target=10.1080/00207179.2018.1543896}{Link 1.} \href{https://www.tandfonline.com/eprint/FsW8JjrVPHmXJ3xuuXZH/full}{Link 2.}
\href{https://www.tandfonline.com/eprint/CHZAMUH3MZUJVRJPTUXN/full?target=10.1080/00207179.2018.1543896}{Link 3.}
}
}
}

\author{
\name{David~Gómez-Gutiérrez\textsuperscript{a,b}, Carlos~Renato~Vázquez\textsuperscript{b}\thanks{CONTACT C.~R. Vázquez: cr.vazquez@itesm.mx}, Sergej~\v{C}elikovsk\'{y}\textsuperscript{c}, Juan Diego~Sánchez-Torres\textsuperscript{d} and Javier~Ruiz-León\textsuperscript{e}}
\affil{\textsuperscript{a}Multi-agent Autonomous Systems Lab, Intel Labs, Intel Tecnología de México, Av. del Bosque 1001, 45019, Zapopan, Jalisco, Mexico; 
\textsuperscript{b}Tecnologico de Monterrey, Escuela de Ingeniería y Ciencias, Av. General Ramón Corona 2514, 45201, Zapopan, Jalisco, Mexico; \textsuperscript{c}The Czech Academy of Sciences, Institute of Information Theory and Automation, Pod vod\'{a}renskou v\v{e}\v{z}\'{\i} 4, 182 08 Prague, Czech Republic, {\tt celikovs@utia.cas.cz}; 
\textsuperscript{d}Research Laboratory on Optimal Design, Devices and Advanced Materials -OPTIMA-, Department of Mathematics and Physics, ITESO, Perif\'erico Sur Manuel G\'omez Mor\'in 8585 C.P. 45604, Tlaquepaque, Jalisco, M\'exico; 
\textsuperscript{e}CINVESTAV Unidad Guadalajara, Av. del Bosque 1145, Zapopan , 45019, Jalisco, Mexico}
}

\maketitle

\begin{abstract}
This paper aims to analyze the stability of a class of consensus algorithms with finite-time or fixed-time convergence for dynamic networks composed of agents with first-order dynamics. In particular, in the analyzed class a single evaluation of a nonlinear function of the consensus error is performed per each node. The classical assumption of switching among connected graphs is dropped here, allowing to represent failures and intermittency in the communications between agents. Thus, conditions to guarantee finite and fixed-time convergence, even while switching among disconnected graphs, are provided. Moreover, the algorithms of the considered class are computationally simpler than previously proposed finite-time consensus algorithms for dynamic networks, which is an essential feature in scenarios with computationally limited nodes and energy efficiency requirements such as in sensor networks. 
Simulations illustrate the performance of the proposed consensus algorithms. In the presented scenarios, results show that the settling time of the considered algorithms grows slower than other consensus algorithms for dynamic networks as the number of nodes increases.
\end{abstract}

\begin{keywords}
Finite-time consensus, Fixed-time consensus, dynamical networks, multi-agent systems, multiple interacting autonomous agents, self-organizing systems
\end{keywords}

\newcounter{tblEqCounter} 

\section{Introduction}
Inspired by the ability of certain social insects to self-organize and mutually cooperate by relying only on neighbor-to-neighbor communication, there has been an increasing interest during the last decade in distributed problems to control the behavior of an agent's network by local interactions. Of particular interest is the consensus problem, which deals with allowing a network of agents to agree on a common value for its internal state by using only communication among neighbors (see e.g. \cite{Olfati2004,Cortes2006,Jiang2009,Chen2013,Lewis2014,Cai2012}). Consensus algorithms have application, for instance, in distributed formation control~\citep{Ren2005,Li2013}, distributed resource allocation~\citep{Xu2017,Xu2017b} and multi-agent rendezvous at the globally optimal point~\citep{Adibzadeh2018}.

The are several published works proposing consensus algorithms in which the agents are first-order integrator systems~\citep{Wang2010}, second-order integrator systems~\citep{Guan2012,Tian2018} or high-order integrator dynamics~\citep{Tian2017,Zuo2018}. Some works consider static communication topologies while others consider dynamic topologies, modeling intermittency in the communications, movement of the agents and the switching between different transmission/reception power levels.

Regarding first-order agents, it is known that if the graph topology is strongly connected, then consensus can be achieved by the standard protocol~\citep{Olfati-Saber2007}. Convergence to the average of the agent's initial states is achieved if the network topology is balanced (identical number of in-neighbors and out-neighbors). For unbalanced graphs, the standard algorithm can be modified by adding a surplus dynamic, to still achieve consensus on the average value~\citep{Cai2012}. These algorithms are linear, and thus the convergence is asymptotic.

Nonlinear protocols have been used to achieve finite-time convergence, in particular, binary protocols have been broadly investigated, achieving consensus to the average value~\citep{Sayyaadi2011,Franceschelli2013}, the average-min-max value~\citep{Cortes2006,Li2014}, the median value~\citep{Franceschelli2017}, and the maximum or minimum value~\citep{Liu2015} of the agent's initial conditions. Continuous finite-time protocols have been introduced in~\cite{Hui2008,Wang2010,Shang2012,Zhu2013}. Moreover, in \cite{Parsegov2013,Zuo2014,Zuo2014a,Ning2018} there have been proposed protocols with fixed-time convergence, i.e., there exists a bound for the convergence time that is independent of the initial conditions. Consensus algorithms where the convergence time is set a priori have been introduced in~\cite{Yong2012,Liu2018}, using linear consensus with a time-varying gain; unfortunately, these methods require that all nodes have a common reference-time which is often restrictive. Multiple variations of the consensus problem have been derived recently. For instance, in~\cite{Meng2016} the multi-scale consensus problem has been addressed, where the nodes agree on a common quantity, but each one with its predetermined scale. In that paper, protocols with asymptotic, finite-time and fixed-time convergence were analyzed for the case of static networks. In \cite{Meng2016a}, the signed-average consensus problem is considered, where the nodes converge to values that are equal in magnitude but may be different in sign. For this problem, fixed-time convergent algorithms were proposed for the static network.

It can be demonstrated that some of the previously mentioned algorithms achieve consensus under dynamic networks switching among connected topologies~\citep{Olfati-Saber2007,Cai2014}, while maintaining finite-time convergence~\citep{Franceschelli2013,Wang2010} or fixed-time convergence~\citep{Zuo2014a}. 
Moreover, some protocols have been shown to achieve consensus in dynamic networks composed by disconnected topologies provided that they form a connected graph in a ``joint sense'', for instance, in~\cite{Chen2015} an algorithm with asymptotic convergence is proposed for the event-triggered consensus problem (where the control action is triggered only when an event is satisfied). In~\cite{Lin2012}, the average consensus problem with time-delay is addressed using asymptotic convergent algorithms. In~\cite{Liu2015}, a discontinuous protocol is analyzed using nonsmooth stability theory.

There exist several works regarding consensus for double integrator agents. Frequently, matching disturbances are considered (for instance, \cite{17, 28,32}). Most of the papers propose finite-time protocols~\citep{36,37,Guan2012,32}. Most of the works consider a fix topology (for instance, \cite{17, 28,32,36,37}). From the current literature, only few works consider dynamic topologies~\citep{27,Guan2012}, where additional conditions about the graph connection must hold. On the other hand, the consensus in high-order agents has been recently considered. In general, agents are described as linear systems (e.g., \cite{4,7}), but there are few works in which agents are non-linear \citep{22,24}. Frequently, the information shared by each agent is its output, then, observers are used to estimate the relative errors and thus evaluate the consensus protocol. Generally, the protocol has the form of linear feedback (for instance, \cite{4,7, 13}), in some cases with an adaptive gain, leading to asymptotic convergence, but non-linear protocols have also been applied~\citep{22,24}. In most of the works the topology is fixed (e.g., \cite{4,7,22}), however, some works consider dynamic topologies \citep{13,16,24} but requiring certain restrictions, for instance, that the topology graphs be jointly connected \citep{Cai2014,16,24}.

\subsection{Paper contribution}

In this work, we revisit the consensus problem for first-order agents considering dynamic topologies. The goal is to analyze a class of consensus algorithms from a common framework, rather than to propose a particular protocol. In the analyzed consensus class, each agent applies a protocol defined as a nonlinear function evaluated on the sum of the errors of all neighboring nodes. By using nonsmooth stability analysis and results from finite-time stability and homogeneity theory, conditions for asymptotic, finite-time and fixed-time consensus are derived for this class of protocols. Three cases are considered: when the communication topology is static; when the communication topology switches among connected graphs; and when the communication topology switches among disconnected graphs. We emphasize that, in the literature, continuous algorithms of this class have only been demonstrated to achieve consensus on static networks. In fact, a detailed comparison to the related protocols existing in the literature is presented in Subsection \ref{subsection-comparison}. The efficacy of the analyzed consensus class on dynamic networks is proven in this paper.

Additionally, it is shown through simulations that the settling time of the analyzed class grows slower when the number of nodes increases than with other related algorithms \citep{Wang2010,Shang2012,Zuo2014}. Another advantage of the analyzed class is that they are computationally simpler than existing finite-time and fixed-time algorithms for dynamic networks \citep{Hui2010,Wang2010,Franceschelli2013,Zuo2014}, which is relevant in applications with energy efficiency requirements and limited computing resources.

The rest of the paper is organized as follows: In Section~\ref{Sec.Preliminaries}, mathematical preliminaries on graph theory and finite-time stability are presented. In Section~\ref{Sec.MainResult}, the main result is derived, and illustrative examples are introduced. In Section~\ref{Sec.IllustrativeExample},
a comparison between consensus algorithms for dynamic networks is presented, analyzing how their convergence time varies as the size of the network increases. Finally, the conclusions and future work are presented in Section~\ref{Sec.Conclusions}.

\section{Preliminaries}
\label{Sec.Preliminaries}

\subsection{Graph Theory}
\label{SubSec.GraphTheory}
The following notation and preliminaries on graph theory are taken mainly from~\cite{godsil2001}.

\begin{definition}
A directed graph (also called digraph) $\mathcal{X}$ consists of a vertex set $\mathcal{V}(\mathcal{X})$ and an edge set $\mathcal{E}(\mathcal{X})$ where an edge (also called node) is an ordered pair of distinct vertices of $\mathcal{X}$. Writing $(i,j)$ denotes an edge with direction from vertex $i$ to vertex $j$. The set of in-neighbors of a vertex $i$ in the graph $\mathcal{X}$ is denoted by $\mathcal{N}^{-}_i(\mathcal{X})=\{j\in\mathcal{V}(\mathcal{X}):(j,i)\in \mathcal{E}(\mathcal{X})\}$. A graph is said to be undirected if $(i,j)\in\mathcal{E}(\mathcal{X})$ implies that $(j,i)\in\mathcal{E}(\mathcal{X})$.
\end{definition}

A path from $i$ to $j$ in a digraph is a sequence of edges $(i,k_1)$ $(k_1,k_2)$ $\cdots$ $(k_{n-1},k_{n})$ $(k_n,j)$ starting in node $i$ and ending in node $j$. A digraph is said to be connected if for every pair of vertices $i,j\in\mathcal{V}(\mathcal{X})$ either there is a path from $i$ to $j$ or a path from $j$ to $i$, otherwise it is said to be disconnected. A subgraph of $\mathcal{X}$ is a graph $\mathcal{Y}$ such that $\mathcal{V}(\mathcal{Y})\subseteq\mathcal{V}(\mathcal{X})$ and $\mathcal{E}(\mathcal{Y})\subseteq\mathcal{E}(\mathcal{X})$. A subgraph $\mathcal{Y}$ such that $\mathcal{V}(\mathcal{Y})\subset\mathcal{V}(\mathcal{X})$ or $\mathcal{E}(\mathcal{Y})\subset\mathcal{E}(\mathcal{X})$ is called a proper subgraph of $\mathcal{X}$.
A subgraph $\mathcal{Y}$ of $\mathcal{X}$ is called an induced subgraph if for any two vertices $i,j\in\mathcal{V}(\mathcal{Y})$, there is an edge $(i,j)\in\mathcal{E}(\mathcal{Y})$ if and only if $(i,j)\in\mathcal{E}(\mathcal{X})$.
An induced subgraph $\mathcal{Y}$ of $\mathcal{X}$ that is connected is called maximal if it is not a proper subgraph of another connected subgraph of $\mathcal{X}$.
A connected induced subgraph of $\mathcal{X}$ that is maximal is called a connected component of $\mathcal{X}$. The edge connectivity $\kappa_1(\mathcal{X})$ of $\mathcal{X}$ is the minimum number of edges that are needed to be removed to decrease the number of connected components.

\begin{definition}
A weighted digraph is a digraph together with a weight function $W:\mathcal{E}(\mathcal{X})\rightarrow \mathbb{R}_+$. The adjacency matrix $A=[a_{ij}]\in\mathbb{R}^{n\times n}$ of a graph with $n$ vertices is a square matrix where $a_{ij}$ corresponds to the weight of the edge $(j,i)$ if $j\in\mathcal{N}^{-}_i(\mathcal{X})$ and $a_{ij}=0$ otherwise.
\end{definition}

 For an undirected graph $a_{ij}=a_{ji}$.
The Laplacian of $\mathcal{X}$ is the matrix  $\mathcal{Q}(\mathcal{X})=\Delta-A$ where $\Delta=diag(d_1\cdots,d_n)$ with $d_i=\sum_{j=1}^na_{ij}$. If the graph $\mathcal{X}$ is connected, then the eigenvalue $\lambda_1(\mathcal{Q})=0$ has algebraic multiplicity one with eigenvector $\mathbf{1}=[1\ \cdots\ 1]^T$, i.e. the right annuler of $\mathcal{Q}(\mathcal{X})$ is $\ker \mathcal{Q}(\mathcal{X})=\{x:x_1=\ldots=x_n\}$. The algebraic connectivity of $\mathcal{X}$ is the second smallest eigenvalue of $\mathcal{Q}$, $\lambda_2(\mathcal{Q})$, which is lower or equal to the edge connectivity of $\mathcal{X}$, i.e. $\lambda_2\leq \kappa_1(\mathcal{X})$~\citep{godsil2001}.

\subsection{Finite-time and Fixed-time Stability}

Some results on homogeneity theory~\citep{Hermes1991, Rosier1992, Kawski1995}, finite-time~\citep{Bhat2000, Bhat2005} and fixed-time stability~\citep{Andrieu2008,Polyakov2016}, which will be useful in the exposition later on, are recalled in this subsection. Let us first present some definitions.

\begin{definition} \citep{Perruquetti2008}
Let $f:\mathbb{R}^n\rightarrow\mathbb{R}^n$ be a piecewise continuous function
with $f(0)=0$, $\psi (t,x_{0})$ is said to be a right-maximally defined solution
of
\begin{equation}
\label{Eq.NonLinear}
\dot{x}=-f(x)
\end{equation}
if $\psi (t,x_{0})$ is such that

$$\frac{{ d\psi (t,x_{0})}}{{ d}t}=-f(\psi (t,x_{0})), ~~\psi
(0,x_{0})=x_{0}, \ \ \ \ \forall t \in [t_0, T_{m}(x_{0})), \forall x_{0}\in {\cal D} \setminus \{ 0
\}
$$
where $T_{m}(x_{0})\in (0,\infty) $ is the maximal possible real number with the above property, or plus infinity. Moreover,~\eqref{Eq.NonLinear} is said to have a unique solution in forward time if for any $x_0\in\mathbb{R}^n$ and two right-maximally defined solutions of~\eqref{Eq.NonLinear}, $\psi (t,x_{0})$ and $\phi (t,x_{0})$ defined on $[t_0,T_m^\psi]$ and $[t_0,T_m^\psi]$, respectively, there exists $t_0<T_m(x_0)<\min(T_m^\psi(x_0),T_m^\phi(x_0))$ such that $\psi (t,x_{0})=\phi (t,x_{0})$ for all $t\in[t_0,T_m(x_0))$.
\end{definition}

Solutions to~\eqref{Eq.NonLinear} are understood in the sense of Filippov and are assumed to be unique in forward-time.

\begin{definition}
A switched nonlinear system
\begin{equation}
\label{Eq.SLS}
\dot{x}=-f_{\sigma(t)}(x)
\end{equation}
is defined by the tuple $\langle\mathcal{P},\sigma\rangle$ where $\mathcal{P}=\{f_1,\ldots,f_m\}$ is a family of nonlinear vector fields such that $\dot{x}=-f_k(x)$, $k\in\{1,\ldots,m\}$, has a unique solution in forward time (understood in the sense of Filippov) and $\sigma:\mathbb{R}_+\rightarrow\{1,\ldots,m\}$ is the switching signal defining the active vector field, such that $f_{\sigma(t)}(x)=f_k(x)$ whenever $\sigma(t)=k$, with the property that only a finite number of switchings occur in any finite interval, i.e. Zeno behavior is excluded.
\end{definition}

The solution $\psi (t,x_{0})$ of~\eqref{Eq.SLS} is absolutely continuous and it is such that, if $\sigma(t)=k$, $\forall t\in[t_i,t_j)$ then $\psi (t,x_0)=\psi_k(t-t_i,x(t_i))$, $\forall t\in[t_i,t_j)$ where $\psi_k(t-t_i,x(t_i))$ is the unique solution in forward-time of $\dot{y}(\hat{t})=-f_k(y(\hat{t}))$ with $\hat{t}=t-t_i$ and initial condition $y(\hat{t}_0)=x(t_i)$.

\begin{definition}
\citep{Bhat2005}
The origin of~\eqref{Eq.SLS} is called finite-time convergent if there exists an open neighborhood $\mathcal{M} \subseteq \mathbb{R}^n$ around the origin and a function $T: \mathcal{M}\setminus \{0\} \rightarrow (0, \infty)$, called the settling-time function, such that for every $x_0\in \mathcal{M}\setminus \{0\}$, the solution $\psi (t,x_{0})$ is defined on $[0,T(x_0))$, $\psi (t,x_{0})\in \mathcal{M}\setminus \{0\}$ for all $t\in[0,T(x_0))$ and $\lim_{t \to T(x_{0})}\psi (t,x_{0})=0$. Because of the uniqueness of $\psi(t,x_{0})$, it follows that $T (x_0) = \mathrm{min}\{t \in \mathbb{R}_+| \psi(t, x_0) = 0\}$. Furthermore, the origin is said to be finite-time stable if it is stable and finite-time convergent. Similarly, the origin is said to be globally finite-time stable if it is finite-time stable with $\mathcal{M} =\mathbb{R}^n$.\\
The origin is called fixed-time convergent for~\eqref{Eq.NonLinear} if it is finite-time convergent and $\forall x_0\in\mathbb{R}^n$ the settling time $T(x_0)$ is bounded by some $T_{max}>0$. Furthermore, the origin is said to be fixed-time stable if it is stable and fixed-time convergent.
\end{definition}

The following definitions introduce the concept of homogeneity in functions and vector fields, which will be used for finite-time and fixed-time stability analysis.

\begin{definition}
\label{Def.Homogeneous}
\citep{Bhat2005}
A function $g:\mathbb{R}^n\rightarrow\mathbb{R}$ is called homogeneous of degree $l$ with respect to the ``standard dilation" $\Delta_\lambda(x)=\lambda x$ if and only if
 $$g(\lambda x)=\lambda^l g(x)$$
for all $\lambda>0$.\\
A vector field $f(x)$, where $x\in\mathbb{R}^n$, is homogeneous of degree $d$ with respect to the standard dilation if $$f(x)=\lambda^{-(d+1)}f(\lambda x).$$ \end{definition}

\begin{definition}
\label{Def:Bilimit}
\citep{Andrieu2008,Polyakov2016}
A function $g:\mathbb{R}^n\rightarrow\mathbb{R}$, such that $g(0)=0$, is said to be homogeneous in the $\lambda_0-$limit with degree $d_{\lambda_0}$ if the function $g_{\lambda_0}:\mathbb{R}^n\rightarrow \mathbb{R}$, defined as
$$
g_{\lambda_0}(x)=\lim_{\lambda\rightarrow\lambda_0}\lambda^{-d_{\lambda_0}}g(\lambda x),
$$
is homogeneous of degree $d_{\lambda_0}$ with respect to the standard dilation.\\
A vector field $f:\mathbb{R}^n\rightarrow\mathbb{R}^n$ is said to be homogeneous in the $\lambda_0-$limit with degree $d_{\lambda_0}$ if the vector field $f_{\lambda_0}:\mathbb{R}^n\rightarrow \mathbb{R}^n$, defined as
\begin{equation}
\label{Eq:Hom}
f_{\lambda_0}(x)=\lim_{\lambda\rightarrow\lambda_0}\lambda^{-(d_{\lambda_0}+1)}f(\lambda x),
\end{equation}
is homogeneous of degree $d_{\lambda_0}$ with respect to the standard dilation.
\end{definition}

The following results provide sufficient conditions for finite-time stability and fixed-time stability, respectively.

\begin{theorem}
\label{Th.FTStability}
\citep[Theorem~7.1]{Bhat2005} Let $f(x)$, with $x\in\mathbb{R}^n$, be an homogeneous vector field of degree $q$ with respect to the standard dilation. 
Then the origin of $\dot{x}=-f(x)$ is globally finite-time stable if and only if it is globally asymptotically stable and $d<0$, where $d$ is the homogeneity degree of $f(x)$.
\end{theorem}

\begin{theorem}
\label{Th.ComHom}
\citep[Theorem~7.4]{Bhat2005}  Suppose $f(x) = f_1(x) + \ldots + f_k(x)$, where $f(0)=0$ and for each $i =
1, \ldots, k$, the vector field $f_i(x)$ is continuous, homogeneous of degree $d_i$
with respect to the standard dilation and $d_1 < \cdots < d_k$. If the origin is a finite-time-stable equilibrium
under $f_1(x)$, then the origin is a finite-time-stable equilibrium under $f(x)$.
\end{theorem}

The results in~\cite{Bhat2005} required continuous vector fields. This restriction was eliminated in~\cite{Orlov2004,Levant2005} and extended for the finite-time stability analysis of switched systems and differential inclusions. Thus, Theorem~\ref{Th.FTStability} and Theorem~\ref{Th.ComHom} hold even if $f(x)$ is not continuous.

\begin{theorem}
\label{Th:Fixed}
\citep{Andrieu2008,Polyakov2016}
Let the vector field $f:\mathbb{R}^n\rightarrow\mathbb{R}^n$ be homogeneous in the $0-$limit with degree $d_0<0$ and homogeneous in the $+\infty$-limit with degree $d_\infty>0$. If for the dynamic systems $\dot{x}=-f(x)$, $\dot{x}=-f_0(x)$ and $\dot{x}=-f_\infty(x)$ the origin is globally asymptotically stable (where $f_0$ and $f_{\infty}$ are obtained from~\eqref{Eq:Hom} with $\lambda_0=0$ and $\lambda_0=+\infty$, respectively), then the origin of $\dot{x}=-f(x)$ is a globally fixed-time stable equilibrium.
\end{theorem}

The following lemma introduces vector fields that guarantee finite-time and fixed-time stability.

\begin{lemma}\label{lem:control} The origin of the nonlinear system~\eqref{Eq.NonLinear} is globally
\begin{itemize}
    \item finite-time stable if
    \begin{equation}
\label{Eq.SignFunction}
f(x)=k\sign(x), \ \ \ \text{with } k>0,
\end{equation}
\item finite-time stable if
\begin{equation}
\label{Eq.FiniteTimeFunction}
f(x)=k\lfloor x \rceil^\alpha, \ \ \ \text{with } k>0 \text{ and } \alpha\in(0,1),
\end{equation}
\item fixed-time stable if
\end{itemize}
\begin{equation}
\label{Eq.FixedTimeFunction}
f(x)=k_1\lfloor x \rceil^p+k_2\lfloor x \rceil^{q}, \ \ \ \text{with } q>1>p\geq0, \ \ \ k_1,k_2>0,
\end{equation}
where $\lfloor x \rceil^\alpha=|x|^\alpha \sign(x)$ and
$$
\sign(x)=
\left\lbrace
\begin{array}{rcc}
1 & \text{if} & x>0 \\
0 & \text{if} & x=0 \\
-1 & \text{if} & x<0.
\end{array}
\right.
$$
\end{lemma}

\subsection{On finite-time and fixed-time consensus for first-order agents}\label{subsection-comparison}

\begin{definition}
A switched dynamic network (or simply dynamic network) $\mathcal{X}_{\sigma(t)}$ is described by the tuple $\mathcal{X}_{\sigma(t)}=\langle\mathcal{F},\sigma\rangle$ where $\mathcal{F}=\{\mathcal{X}_1,\ldots,\mathcal{X}_m\}$ is a collection of undirected graphs having the same vertex set $\mathcal{V}(\mathcal{X}_{\sigma(t)})$ and $\sigma:[t_0,\infty)\rightarrow \{1,\ldots m\}$ is a switching signal that determines the topology of the dynamic network at each instant of time, i.e. $\mathcal{X}_{\sigma(t)}=\mathcal{X}_i$ when $\sigma(t)=i$.\\
Furthermore, each vertex $i\in \mathcal{V}(\mathcal{X}_{\sigma(t)})$ is associated to an agent $i$, with a first-order integrator dynamics:
\begin{equation}
\label{Consensus}
\dot{x}_i(t)=u_i(t)
\end{equation}
where $x_i(t)\in \mathbb{R}$ is the state of the $i-th$ agent and $u_i(t)$ is its consensus protocol. The network is said to achieve consensus if the evolution of the network converges to $x_1=\cdots=x_n$.

At a given time $t\geq0$, an agent $i$ has access to the state of its in-neighbors agents $\mathcal{N}^{-}_i(\mathcal{X}_{\sigma(t)})$.
\end{definition}

In this paper, we assume that $\sigma(t)$ is exogenously generated and that there is a minimum dwell time $\tau_{min}$ between consecutive switchings in such a way that Zeno behavior in the network's dynamic is excluded, i.e., there is a finite number of switchings in any finite time interval.

Consider a continuous nonlinear function $f:\mathbb{R}\rightarrow\mathbb{R}$, with $f(0)=0$, such that the origin of $\dot{x}=-f(x(t))$ is globally asymptotically stable. Then, two directions (approaches) can be taken to derive nonlinear consensus algorithms, where the protocol for each node $i$ is based on the states of the nodes in its in-neighbor set $\mathcal{N}^{-}_{i}(\mathcal{X}_{\sigma(t)})$. On the one hand, an approach is defined by applying the function $f(x)$ on the error described by each neighboring node, i.e.,
\begin{equation}
\label{Eq.NonlinearConsensusA}
u_i(t)=\sum_{j\in\mathcal{N}^{-}_i(\mathcal{X}_{\sigma(t)})}a_{ij}f(e_{ij}), \ \ \  e_{ij}=x_j(t)-x_i(t).
\end{equation}

On the other hand, another approach is defined by applying the $f(x)$ on the sum of the errors of all neighboring nodes, i.e.,
\begin{equation}
\label{Eq.NonlinearConsensusB}
u_i(t)=f (e_i), \ \ \ \ e_i=\sum\limits_{j\in\mathcal{N}^{-}_i(\mathcal{X}_{\sigma(t)})}a_{ij}e_{ij}=\sum\limits_{j\in\mathcal{N}^{-}_i(\mathcal{X}_{\sigma(t)})}a_{ij}(x_j(t)-x_i(t)).
\end{equation}

\begin{definition}\label{def:dir}
Approaches~\eqref{Eq.NonlinearConsensusA} and~\eqref{Eq.NonlinearConsensusB} are said to be direction~\eqref{Eq.NonlinearConsensusA} and direction~\eqref{Eq.NonlinearConsensusB}, respectively.
\end{definition}

Based on standard feedback controllers presented in Lemma~\ref{lem:control} and considering the directions~\eqref{Eq.NonlinearConsensusA} and~\eqref{Eq.NonlinearConsensusB} in the sense of Definition~\ref{def:dir}, Table~\ref{Tab:DiferentConsensus} provides particular consensus algorithms of the form $f(x)=g_r(x)$, $r=1\ldots4$. 

\setcounter{tblEqCounter}{\theequation}
\begin{table}[H]
\scriptsize
    \centering
    \begin{tabular}{|c|l|l|}
    \hline
        $f(x)$ & Direction~\eqref{Eq.NonlinearConsensusA} & Direction~\eqref{Eq.NonlinearConsensusB}\\
        \hline
        $g_1(x)=kx$ & $u_i=k\sum\limits_{j\in\mathcal{N}^{-}_i(\mathcal{X}_{\sigma(t)})} a_{ij}e_{ij}$ \hfill\numberTblEq{Eq:Linear} &  $u_i=k e_i$ \hfill \\
        $g_2(x)=k\sign(x)$, & $u_i=k\sum\limits_{j\in\mathcal{N}^{-}_i(\mathcal{X}_{\sigma(t)})} a_{ij}\sign(e_{ij})$ \hfill\numberTblEq{Eq:ConsSign} & $u_i=k\sign(e_i)$ \hfill \numberTblEq{Eq:ConsSignMine}\\
        $g_3(x,\alpha)=k\lfloor x \rceil^\alpha$ & $u_i=k\sum\limits_{j\in\mathcal{N}^{-}_i(\mathcal{X}_{\sigma(t)})} a_{ij}\lfloor e_{ij} \rceil^\alpha$ \hfill\numberTblEq{Eq.OtherFTConsesusAlg} & $u_i=k\lfloor e_i \rceil^\alpha$\hfill\numberTblEq{Eq:ConsFiniteMine}\\
        $g_4(x,p,q)=k_1\lfloor x \rceil^p+k_2\lfloor x \rceil^q$ & $u_i=\sum\limits_{j\in\mathcal{N}^{-}_i(\mathcal{X}_{\sigma(t)})}a_{ij}(k_1\lfloor e_{ij} \rceil^p+k_2\lfloor e_{ij} \rceil^q)$\hfill\numberTblEq{Eq.FixedConsesusAlg} & $u_i=k_1\lfloor e_i \rceil^p+k_2\lfloor e_i \rceil^q$\hfill\numberTblEq{Eq:ConsFixedMine}\\   
\hline
    \end{tabular}
    \caption{Some examples of consensus algorithms derived following direction~(\ref{Eq.NonlinearConsensusA}) and direction~(\ref{Eq.NonlinearConsensusB}).}
    \label{Tab:DiferentConsensus}
\end{table}
\setcounter{equation}{\thetblEqCounter} 

\begin{remark}\label{rem:comb}
The basic stabilizing functions given in Table~\ref{Tab:DiferentConsensus} can be combined to generate consensus algorithms following either direction~(\ref{Eq.NonlinearConsensusA}) or direction~(\ref{Eq.NonlinearConsensusB}), by taking  $f(x)=l_1(x)g_1(x)+l_2(x)g_2(x)+l_3(x)g_3(x,\alpha)+l_4(x)g_4(x,p,q)$ where $l_r(x)$, $r=1\ldots4$ are nonnegative piecewise constant functions (when $l_r(x)$ is constant we simply write $l_r$) not all zero at the same time. In this paper, our focus is on protocols obtained from $f(x)$ following direction~\eqref{Eq.NonlinearConsensusB}; and we will derive conditions on $f(x)$ such that a finite-time or a fixed-time consensus algorithm for dynamic networks is obtained.
\end{remark}

In the following, some common consensus protocols proposed in the literature will be presented as particular cases of the combination given in Remark~\ref{rem:comb}.

The standard consensus algorithm~\eqref{Eq:Linear} proposed in~\cite{Olfati-Saber2007} is derived from $f(x)=g_1(x)=kx$, since, for this case, $f(\cdot)$ is a linear function then direction~(\ref{Eq.NonlinearConsensusA}) and direction~(\ref{Eq.NonlinearConsensusB}) are equivalent and its convergence is asymptotic. 

Regarding finite consensus following direction~(\ref{Eq.NonlinearConsensusA}), in~\cite{Hui2010,Sayyaadi2011,Chen2011}
the discontinuous consensus algorithm in~\eqref{Eq:ConsSign} was shown to achieve finite-time convergence, while~\cite{Franceschelli2013} showed that consensus is also achieved in the presence of disturbances and dynamic networks switching among connected topologies. The protocol~\eqref{Eq.OtherFTConsesusAlg} was shown in~\cite{Hui2008,Xiao2009} to be a finite-time consensus algorithm for static networks, while~\cite{Wang2010} showed that it provides finite-time convergence for dynamic networks switching among connected topologies. In~\cite{Liu2016} a protocol switching between~\eqref{Eq:ConsSign} and ~\eqref{Eq.OtherFTConsesusAlg} was proposed, exhibiting finite-time convergence, whereas in~\cite{Cao2014} a finite-time consensus algorithm was obtained by switching between~\eqref{Eq:Linear} and~\eqref{Eq.OtherFTConsesusAlg}. In \cite{Wang2018} finite-time consensus for switching dynamic networks was obtained from $f(x)=l_1g_3(e_{ij},\alpha)+l_2g_1(x)$.

Regarding fixed-time consensus following direction~(\ref{Eq.NonlinearConsensusA}). In~\cite{Parsegov2013,Zuo2014}, the protocol~\eqref{Eq.FixedConsesusAlg} was shown to be a fixed-time consensus algorithm for static networks, later, fixed-time convergence in networks switching among connected topologies was demonstrated in~\cite{Zuo2014a}. In~\cite{Hong2017} different consensus protocols for static networks were proposed derived from $f(x)=l_1g_2(x)+l_2g_3(x,\alpha)+l_3g_4(x,p,q)$ with $l_1,l_2,l_3\geq0$ and $l_2,l_3$ not both zero. In~\cite{Sharghi2016} a fixed-time consensus based on $f(x)=l_1g_1(x)+l_2g_4(x,p,q)$ with $l_1,l_2>0$, was proposed for the leader-follower consensus problem in static networks.

Deriving finite and fixed-time consensus following direction~(\ref{Eq.NonlinearConsensusB}) has been less explored and, as shown below, its analysis has been mainly focused on static networks. In~\cite{Cortes2006} it was shown that~\eqref{Eq:ConsSignMine} is a finite-time consensus for static networks, while~\cite{Franceschelli2015} showed that finite-convergence is maintained in the presence of disturbances and under dynamic networks switching among connected topologies. Furthermore, in~\cite{Li2014,Liu2015} it was shown that~\eqref{Eq:ConsSignMine} is still a consensus algorithm even when switching among disconnected topologies. Although~\eqref{Eq:ConsFiniteMine} and~\eqref{Eq:ConsFixedMine} have been shown to achieve finite-time and fixed-time convergence in~\cite{Xiao2009,Wang2010,Shang2012,Gomez2018} and~\cite{Zuo2014a}, respectively, the results of these papers are restricted to static connected networks. In~\cite{Defoort2015} a fixed-time consensus for the leader-follower consensus problem was presented for static networks.

Recently, in~\cite{Ning2017b,Ning2018} a discontinuous consensus algorithm for static networks was proposed showing that if the protocol is the sum of the linear protocol in~\cite{Olfati-Saber2007} and \eqref{Eq:ConsSignMine} finite-time consensus is obtained; whereas if the protocol is the sum of~\eqref{Eq:ConsFixedMine} and~\eqref{Eq:ConsSignMine} fixed-time consensus is obtained.

A comparison among the different papers addressing the finite-time and the fixed-time consensus problem following direction ~\eqref{Eq.NonlinearConsensusB} is summarized in Table~\ref{Tab:Comparison}. As it can be noted, for papers based on $g_3(x,\alpha)$ and $g_4(x,p,q)$ no formal proofs have been presented in the literature to show that these methods can be applied for networks switching among connected graphs nor for networks forming a jointly connected graph, the main reason is that their analysis is based on Lyapunov functions candidates that are graph dependent. Thus, the argument of a common Lyapunov function cannot be made in their case to show convergence in switched dynamic networks.

\begin{table}[H]
    \small
    \centering
    \begin{tabular}{|c|c|c|c|c|c|c|c|}
    \hline
        Reference & $f(x)$, $l_1,l_2,l_3>0$ & Network Type & Convergence\\
        \hline
        \cite{Cortes2006} & $g_2(x)$ & Static & finite-time\\
        \cite{Xiao2009} & $g_3(x,\alpha)$ & Static & finite-time\\
        \cite{Wang2010} & $g_3(x,\alpha)$ & Static & finite-time\\
        \cite{Shang2012} & $g_3(x,\alpha)$ & Static & finite-time\\
        \cite{Li2014} & $g_2(x)$ & \textit{JC} & finite-time\\
        \cite{Zuo2014a} & $g_4(x,p,q)$ & Static & fixed-time\\
        \cite{Franceschelli2015} & $g_2(x)$ & \textit{SC} & finite-time\\
        \cite{Liu2015} & $g_2(x)$ & \textit{JC} & finite-time\\
        \cite{Defoort2015} & $l_1g_1(x)+l_2g_4(x,2,0)$ & Static & fixed-time\\
        \cite{Tu2017}  & $g_3(x,\alpha)$  &Static & finite-time\\
        \cite{Shang2017} & $l_1g_2(x)+l_2g_4(x,p,q)$  & Static & fixed-time \\
        \cite{Ning2017b} & $l_2g_2(x)+l_4g_4(x)$ & Static & finite-time\\
        \cite{Ning2018} & $l_1g_1(x)+l_2g_2(x)$ & Static & finite-time\\
        \cite{Ning2018} & $l_1g_1(x)+l_2g_2(x)+l_3g_4(x,p,q)$  & Static & fixed-time \\
        \hline
        
    \end{tabular}
    \caption{Comparison of papers presenting finite-time and fixed-time consensus algorithms following direction~\eqref{Eq.NonlinearConsensusB}. Here \textit{SC} stands for networks switching among connected topologies and \textit{JC} stands for networks forming a jointly connected graph and the functions $g_i(\bullet)$, $i=1,\ldots,4$ are defined in Table~\ref{Tab:DiferentConsensus}.}
    \label{Tab:Comparison}
\end{table}

\begin{remark}
The aim of this paper is to analyze finite-time and fixed-time consensus algorithms derived following direction~(\ref{Eq.NonlinearConsensusB}) to show, by using nonsmooth stability analysis, that finite-time and fixed-time consensus is achieved also in dynamic networks. We analyze dynamic networks switching among connected topologies as well as dynamic networks composed of disconnected topologies but forming a connected graph in a ``joint sense''.
\end{remark}

The considered class includes asymptotic, finite-time and fixed-time convergent protocols. Notice that, if the initial conditions are known to belong to a bounded set, fixed-time algorithms may not represent an advantage over finite-time or asymptotic algorithms, because the gain of the latter ones may be selected to provide a desired settling time for any initial condition in the set. On the other hand, under the same topology and initial conditions, fixed-time protocols may require more energy than finite-time protocols to achieve consensus at the same convergence time (which can be seen in the benchmark herein presented), similarly, finite-time protocols may require more energy than asymptotic protocols. Of course, fixed-time consensus algorithms have a great advantage when the initial conditions are unknown and unbounded, because they guarantee the existence of a bound for the convergence time.

\section{Main Result}
\label{Sec.MainResult}

If a consensus algorithm based on direction \eqref{Eq.NonlinearConsensusB} is applied to a dynamic network, then the closed-loop behavior can be compactly represented using a vectorial notation. For this, let $x=[x_1,\ldots,x_n]^T$ be the state vector of the agents of the dynamic network. Let $e_i=\sum_{j\in\mathcal{N}^{-}_j(\mathcal{X}_{\sigma(t)})}a_{ij}(x_j-x_i)$ be the consensus error at node $i$ and let $e=[e_1,\ldots,e_n]^T$ be the consensus error vector. Then, it can be shown that $e=-\mathcal{Q}(\mathcal{X}{_{\sigma{(t)}}}) x$, thus the network's behavior is
\begin{equation}
\label{ConsDynamic}
\dot{x}=-F(\mathcal{Q}(\mathcal{X}{_{\sigma{(t)}}}) x)=F(e)\text{ where }
F(e)=\left[\begin{array}{c}
f(e_1)\\
\vdots \\
f(e_n)
\end{array}\right]. 
\end{equation}

Note that $\mathcal{X}_{\sigma(t)}$ is a switched dynamic network and~\eqref{ConsDynamic} is a switched nonlinear system. 

Given a dynamic network, in this section it is proved that consensus is reached by using the direction~(\ref{Eq.NonlinearConsensusB}). Namely, taking
\begin{equation}
\label{Eq:ConsensusProtocol}
\dot{x}_i=u_i\ \ \ \ \ u_i=f(e_i), \ \ \ \ \ e_i=\sum_{j\in\mathcal{N}^{-}_i(\mathcal{X}_{\sigma(t)})}a_{ij}(x_j(t)-x_i(t))
\end{equation}
with $f:\mathbb{R}\rightarrow\mathbb{R}$ such that $f(0)=0$ and the origin is a globally asymptotically stable equilibrium of the system \eqref{Eq.NonLinear}. Moreover, convergence is guaranteed not only for static or dynamic networks switching among connected graphs, but also for dynamic networks switching among disconnected graphs, provided that $\exists\tau<\infty$ such that for any time interval $[\bar{t},\bar{t}+\tau]$ the graph $\bar{\mathcal{X}}$ with vertex set
$\mathcal{V}(\bar{\mathcal{X}})=\mathcal{V}(\mathcal{X}_{\sigma(t)})$ and edge set
$\mathcal{E}(\bar{\mathcal{X}})=\mathcal{E}(\mathcal{X}_{\sigma(t_i)})\cup \cdots \cup \mathcal{E}(\mathcal{X}_{\sigma(t_k)})$ is connected, where $t_i,\ldots,t_k$ are the successive switching times in the time interval $[\bar{t},\bar{t}+\tau]$.
Furthermore, it is shown that if the function $f(\bullet)$ is such that the origin is a globally finite-time (resp. fixed-time) stable equilibrium of the system \eqref{Eq.NonLinear}, and $f(\bullet)$ satisfies the conditions of Theorem~\ref{Th.FTStability} (resp. Theorem~\ref{Th:Fixed}), then the network's closed-loop system reaches consensus in finite-time (resp. fixed-time).

\begin{remark}
Consensus algorithms obtained by using direction~(\ref{Eq.NonlinearConsensusB}) are computationally less expensive than previously proposed finite-time consensus algorithms for dynamic networks that use direction~(\ref{Eq.NonlinearConsensusA}), particularly~\cite{Wang2010} and~\cite{Zuo2014a}. In detail, direction~(\ref{Eq.NonlinearConsensusB}) only requires a single evaluation of the nonlinear function $f(\bullet)$ for each node, whereas direction~(\ref{Eq.NonlinearConsensusA}) requires a number of evaluations of $f(\bullet)$ for each node equal to the number of its in-neighbors, a number that grows in highly connected topologies.
\end{remark}

\subsection{Consensus over static networks}

Assuming that the communication topology is static and connected, the asymptotic convergence to the consensus state of the standard consensus algorithm is shown in this subsection by using the Lyapunov theory. Afterwards, it is shown, by using homogeneity results~\citep{Hermes1991, Rosier1992, Kawski1995,Andrieu2008,Polyakov2016}), that if $f(\bullet)$ satisfies the conditions of Theorem~\ref{Th.FTStability} (resp. Theorem~\ref{Th:Fixed}) then the consensus algorithm is finite-time (resp. fixed-time) convergent.

The convergence to the consensus state will be demonstrated by showing that, the nonsmooth function
\begin{equation}
\label{Eq.LyapunoFunc}
V(x)=\max\{x_1,\ldots,x_n\}-\min\{x_1,\ldots,x_n\},
\end{equation}
also introduced in~\cite{Sayyaadi2011} and \cite{Liu2015} to define the discontinuous consensus protocols~\eqref{Eq:ConsSign} and \eqref{Eq:ConsSignMine}, monotonically decreases along the closed-loop system and it converges to zero. For this, a couple of lemmas will introduce some basic properties of $V(x)$. In detail, the forthcoming Lemma \ref{Lem.MinLips} states that $V(x)$ is absolutely continuous and Lipschitz. Later, Lemma \ref{lem:Sergej} demonstrates that $V(x)$ is differentiable almost everywhere. Based on these properties, Theorem \ref{Prop.AsympStability} provides sufficient conditions for the asymptotic stability of the consensus state in the closed-loop system.

\begin{remark}
When using a candidate Lyapunov function that is not everywhere differentiable, the stability analysis requires the use of tools for nonsmooth analysis, for instance, using Dini derivatives. However, if $V(x)$ is Lipschitz continuous, by \cite[Lemma 6.1]{Bacciotti2006}, $V(x(t))$ is nonincreasing along the evolution of the system if $\dot{V}(x(t))\leq0$ almost everywhere. The subsequent analysis is based on this result.
\end{remark}

\begin{lemma}
\label{Lem.MinLips}
$V(x)$ defined in~\eqref{Eq.LyapunoFunc} is Lipschitz continuous.
\end{lemma}
\begin{proof}
We show by induction that $\min(x_1,\ldots,x_n)$ is Lipschitz. Notice that $\min(x_1,x_2)=\frac{1}{2}(x_1+x_2-|x_1-x_2|)$ is Lipschitz as the linear combination and the composition of Lipschitz functions are Lipschitz~\citep[ch.~12]{Eriksson2013}. As an induction step assume that $z_{n-1}=\min(x_1,\ldots,x_{n-1})$ is Lipschitz, then by using the previous argument it follows that $\min(x_1,\ldots,x_n)=\frac{1}{2}(z_{n-1}+x_n-|z_{n-1}-x_n|)$ is Lipschitz. Since $\max(a,b)=-\min(-a,-b)$ then \eqref{Eq.LyapunoFunc} is Lipschitz.\\
Finally, since $x(t)$ is the solution of a differential equation then it is absolutely continuous, thus $g(t)=V(x(t))$ is absolutely continuous if~\eqref{Eq.LyapunoFunc} is locally Lipschitz~\citep[p.~207]{Bacciotti2006}, which has been demonstrated.
\end{proof}

\begin{lemma}
\label{lem:Sergej}
Let $\mathcal{S}$ be the set of all solutions of \eqref{ConsDynamic} and let $g(t)=V(x(t))$ where $x(t)\in\mathcal{S}$ and $V(x)$ is as in~\eqref{Eq.LyapunoFunc}.
Then, for any $t^*$ such that $g(t)$ is not differentiable at $t^*$ there exists $\varepsilon_{t^*}>0$ such that 
$g(t)$ is differentiable $\forall t\in [t^{*}-\varepsilon_{t^{*}} , t^{*})  \cup (t^{*}, t^{*}+\varepsilon_{t^{*}} ]$. 
\end{lemma}

\begin{proof}
Note that $g(t)$ is differentiable if both $x_{max}(t)=\max\{x_1,\ldots,x_n\}$ and  $x_{min}(t)=\min\{x_1,\ldots,x_n\}$ are differentiable. Therefore, to prove the lemma it is sufficient to prove its statement with
$g(t)$ being replaced by $x_{max}(t)$ and then to prove its statement with
$g(t)$ being replaced by $x_{min}(t)$. We will provide such a proof for $x_{max}(t)$ only, since the proof for $x_{min}(t)$ is analogous.

To do so, let us notice that $x_{max}(t)$ is not differentiable at time $t^{*}$ only if $\exists k\ge 2 $ and positive integers $i_{1},\ldots , i_{k}$ such that
$$x_{i_{1}}(t^{*})=\ldots = x_{i_{k}}(t^{*})=x_{max}(t^{*}),~~x_{j}(t^{*})<x_{max}(t^{*})~\forall j\not\in \{i_{1},\ldots , i_{k}\} ,$$
$$
\frac{{\rm d}}{{\rm d}t}x_{i_{1}}(t^{*})\le \ldots \le \frac{{\rm d}}{{\rm d}t}x_{i_{k-1}}(t^{*})
< \frac{{\rm d}}{{\rm d}t}x_{i_{k}}(t^{*}).$$
Then, using a simple continuity argument, there exists $\varepsilon_{t^{*}}^{1}>0 $ such that
$$x_{j}(t)<x_{max}(t)~\forall j\not\in \{i_{1},\ldots , i_{k}\},~\forall t\in [t^{*}-\varepsilon_{t^{*}}^{1} ,  t^{*}+\varepsilon_{t^{*}}^{1}].$$
Moreover, by using a simple Taylor expansion argument it can be proved that there exists $\varepsilon_{t^{*}}^{2} >0 $ such that
$$
x_{i_{k}}(t) > x_{i_{k-1}}(t)\ge \ldots \ge x_{i_{1}}(t),~\forall t\in (t^{*}, t^{*}+\varepsilon_{t^{*}}^{2}].
$$
As a consequence,
$$
x_{max}(t)=\max\{ x_{i_{k}}(t) , x_{i_{k-1}}(t),  \ldots , x_{i_{1}}(t)\}, ~\forall t \in (t^{*}, t^{*}+\varepsilon_{t^{*}} ],~\varepsilon_{t^{*}}:= \min\{\varepsilon_{t^{*}}^{1}, \varepsilon_{t^{*}}^{2}\},
$$
and therefore
$$
x_{max}(t)= x_{i_{k}}(t)> x_{j}(t) ~\forall j\in \{1,2, \ldots , n\}\setminus\{ i_{k}\},~\forall t \in (t^{*}, t^{*}+\varepsilon_{t^{*}} ],~\varepsilon_{t^{*}}:= \min\{\varepsilon_{t^{*}}^{1}, \varepsilon_{t^{*}}^{2}\} ,
$$
which in turn guarantees that $ x_{max}(t) $ is differentiable $\forall t \in  (t^{*}, t^{*}+\varepsilon_{t^{*}} ]$.

On the other hand, by a Taylor expansion argument, there exists $\varepsilon_{t^{*}}^{3} >0 $ such that
$$
x_{i_{1}}(t) \ge x_{i_{2}}(t)\ge \ldots > x_{i_{k}}(t),~\forall t\in [t^{*}-\varepsilon_{t^{*}}^{3}, t^{*}),
$$
and thus
$$
x_{max}(t)= x_{i_{1}}(t)\geq x_{j}(t) ~\forall j\in \{1,2, \ldots , n\}\setminus\{ i_{1}\},~\forall t \in [t^{*}-\varepsilon_{t^{*}}, t^{*}),~\varepsilon_{t^{*}}:= \min\{\varepsilon_{t^{*}}^{1}, \varepsilon_{t^{*}}^{3}\} ,
$$
which guarantees that $ x_{max}(t) $ is differentiable $\forall t \in  [t^{*}-\varepsilon_{t^{*}}, t^{*})$.
Thus, the claim of the lemma is proved.
\end{proof}

\begin{theorem}
\label{Prop.AsympStability}
Consider a dynamic network $\mathcal{X}_{\sigma(t)}=\langle\mathcal{F},\sigma\rangle$ such that $\sigma(t)=r$, $\forall t\geq t_0$, and $\mathcal{X}_r$ is a connected graph. Consider a consensus algorithm defined by direction \eqref{Eq.NonlinearConsensusB} and a continuous function $f:\mathbb{R}\rightarrow\mathbb{R}$ such that the origin is a globally asymptotically stable equilibrium of the system $\dot{x}=-f(x)$. Then, the equilibrium of the network's closed-loop system is globally asymptotically stable.
\end{theorem}

\begin{proof}
Notice that ~\eqref{Eq.LyapunoFunc} is radially unbounded and $V(x)> 0$ if $x\notin \ker \mathcal{Q}(\mathcal{X}_r)$, where $\ker \mathcal{Q}(\mathcal{X}_r)=\{x:x_1=\cdots=x_n\}$ is the set of equilibrium points of the network's closed-loop system \eqref{ConsDynamic}, i.e. the consensus states. Notice that $V(x)=0$ implies $x\in \ker \mathcal{Q}(\mathcal{X}_r)$.

Moreover,~\eqref{Eq.LyapunoFunc} is Lipschitz continuous by Lemma \ref{Lem.MinLips}. Thus, according to \cite[Lemma 6.1]{Bacciotti2006}, $V(x)$ is nonincreasing along the network's closed-loop behavior~\eqref{ConsDynamic} if $\dot{V}(x)\leq0$ for almost every $x\notin\ker \mathcal{Q}(\mathcal{X}_k)$, which will be demonstrated in the sequel.

Now, according to Lemma~\ref{lem:Sergej}, $V(x(t))$ is continuously differentiable except on a set of isolated points $\{t_{1}^*,t_{2}^*,...\}$.
Let $\{ t_{1}^*,t_{2}^*, \ldots \}$ be the set of points where $V(t)$ is not differentiable, then $\forall t\in (t_i^*,t_{i+1}^*)$ , $V(x)$ is differentiable with time derivative, $\dot{V}(x)=(f(e_j)-f(e_k))$, where $x_j=x_{max}$ and $x_k=x_{min}$. Since $x_j=x_{max}$ and $e_j=\sum_{i\in\mathcal{N}^{-}_j(\mathcal{X}_r)}a_{ji}(x_i-x_j)$, then
$\sign(e_j)=-1$ and thus $f(e_j)=-|f(e_j)|$. By using a similar argument, it can be shown that $\sign(e_k)=1$ with $f(e_k)=|f(e_k)|$ and therefore
\begin{equation}
\label{Eq.ProfT11}
\dot{V}=-(|f(e_j)|+|f(e_k)|)\leq0,\ \forall t\in (t_i^*,t_{i+1}^*).
\end{equation}

Next, asymptotic convergence can be proved by using LaSalle's invariance principle~\citep{Khalil2002}. To this aim, let $E=\{x\in \mathbb{R}^n\setminus \ker \mathcal{Q}(\mathcal{X}_r)\,|\,\dot{V}(x)=0 \}$. Let $x(t)\in E$ for a nonzero subinterval $(\hat{t}_i^*,\hat{t}_{i+1}^*)\subseteq(t_i^*,t_{i+1}^*)$. Thus, according to \eqref{Eq.ProfT11} $\dot{V}(x)=0$ implies $f(e_j)=f(e_k)=0$, which implies $e_j=e_k=0$ because of the theorem's conditions on $f(\bullet)$.
Now, since $x_j$ is the maximum, $e_j=\sum_{i\in\mathcal{N}^{-}_j(\mathcal{X}_k)}a_{ji}(x_i-x_j)=0$ implies $x_j=x_i$
$\forall i\in\mathcal{N}^{-}_j(\mathcal{X}_k)$ and $\forall t\in(\hat{t}_i^*,\hat{t}_{i+1}^*)$.
Furthermore, $e_i=0$ for all $i\in\mathcal{N}^{-}_j(\mathcal{X}_r)$ (otherwise, if $e_i>0$ then $f(e_i)>0$ and thus $\dot{x}_i(t)<0$, which implies $x_i(t-\delta)>x_j(t-\delta)$ for a small enough $\delta>0$ with $t-\delta\in (\hat{t}_i^*,\hat{t}_{i+1}^*)$, i.e. a contradiction; by an analogous reason $e_i<0$ cannot occur), which implies $x_j=x_i=x_l$, $\forall i\in\mathcal{N}^{-}_j(\mathcal{X}_k),\, \forall l\in\mathcal{N}^{-}_i(\mathcal{X}_k)$. By iterating this reasoning it can be concluded that $x_j=x_p$ for any node $p$ such that there exists a path from $p$ to $j$. In particular, since $\mathcal{X}_r$ is connected, there exists a path from $k$ to $j$, hence $\max(x_1,\ldots,x_n)=x_j=x_k=\min(x_1,\ldots,x_n)$, which clearly implies that $x_1=\ldots=x_n$. Thus, the equality holding in~\eqref{Eq.ProfT11} for a nonzero interval implies that consensus is achieved.
Since $V(x(t))$ is absolutely continuous along the closed-loop trajectory \eqref{ConsDynamic} and according to Lemma~\ref{lem:Sergej} it is differentiable almost everywhere, then by~\eqref{Eq.ProfT11} and LaSalle's invariance principle~\citep{Khalil2002} $\dot{V}<0$ for almost every $t$ such that $V(t)\neq0$, therefore $V(t)$ is decreasing excepting at the consensus states, i.e. the closed-loop system asymptotically converges to the consensus state.
\end{proof}

In the following theorem, additional conditions are given for finite-time and fixed-time convergence of the consensus protocols.

\begin{theorem}
\label{Th.FixedFinite}
Consider a dynamic network $\mathcal{X}_{\sigma(t)}=\langle\mathcal{F},\sigma\rangle$ such that $\sigma(t)=r$, $\forall t\geq t_0$, and $\mathcal{X}_r$ is a connected graph. Consider a consensus algorithm defined by direction \eqref{Eq.NonlinearConsensusB} and a continuous function $f:\mathbb{R}\rightarrow\mathbb{R}$ such that the origin is a globally asymptotically stable equilibrium of the system $\dot{x}=-f(x)$. Then, the consensus algorithm
\begin{enumerate}
\item is a finite-time consensus algorithm if the vector field $f(x)$ can be written as $f(x)=f_1(x)+\ldots+f_k(x)$ and for each $i=1, \ldots, k$, the vector field $f_i(x)$ is homogeneous of degree $d_i$ with respect to the standard dilation and $d_1 <\cdots < d_k$ with $d_1<0$. \label{Th.Item1}
\item is a finite-time consensus algorithm if $f(x)$ is a piecewise function such that there exist a nonzero constant $b$ where $\forall x\in\{x:|x|<b\}$, $f(x)$ satisfies condition~(\ref{Th.Item1}). \label{Th.Item3} 
\item is a fixed-time consensus algorithm if the vector field $f(x)$ is homogeneous in the $0-$limit with degree $d_0<0$, homogeneous in the $+\infty-$limit with degree $d_\infty>0$ and the origin is a globally asymptotically stable equilibrium of the dynamic systems $\dot{x}=-f_0(x)$ and $\dot{x}=-f_\infty(x)$ (where $f_0$ and $f_{\infty}$ are obtained from~\eqref{Eq:Hom} with $\lambda_0=0$ and $\lambda_0=+\infty$, respectively).
\label{Th.Item2}

\end{enumerate}
\end{theorem}
\begin{proof}
Theorem~\ref{Prop.AsympStability} states that the closed-loop behavior \eqref{ConsDynamic} converges asymptotically to its equilibrium. Thus, based on Theorem~\ref{Th.FTStability} and Theorem~\ref{Th.ComHom}, statement~\ref{Th.Item1} holds since the condition in statement~\ref{Th.Item1} implies that the vector field $F(e)$ $e=-\mathcal{Q}(\mathcal{X}_k)x$, can be written as $F(e)=F_1(e)+\cdots+F_k(e)$ such that for each $i=1, \ldots, k$, $F_i(e)$ is homogeneous of degree $d_i$ with respect to the standard dilation, where $d_1<0$ is the smallest degree. This property holds given that if $f_i(x)$ is homogeneous of degree $d_i$ then $F(-\mathcal{Q}(\mathcal{X}_k)\lambda x)=\lambda^{(d_i+1)}F_i(-\mathcal{Q}(\mathcal{X}_k) x)$ for all $\lambda>0$. To show that item (\ref{Th.Item3}) holds, notice that asymptotic convergence to the consensus state implies that after a finite-time the trajectory $x(t)$ will belong to the nonempty set $\{x:\parallel \mathcal{Q}(\mathcal{X}_k)x\parallel_\infty<b \}$ from which the conditions of item (\ref{Th.Item3}) are satisfied, thus achieving finite-time convergence to a consensus state where $\mathcal{Q}(\mathcal{X}_k)x=0$.

Now, let us demonstrate the statement~\ref{Th.Item2}. Consider the closed-loop behavior \eqref{ConsDynamic} and a parameter $d_{\lambda_0}\in\mathbb{R}$. Next, by considering $F_{\lambda_0}(\bullet)$ and $f_{\lambda_0}(\bullet)$ as defined in \eqref{Eq:Hom}, it follows
\begin{equation}\label{eq:aux1}
F_{\lambda_0}(e)=\lim_{\lambda\rightarrow\lambda_0}\lambda^{-(d_{\lambda_0}+1)}F(e)=\left[\begin{array}{c}
-\lim_{\lambda\rightarrow\lambda_0}\lambda^{-(d_{\lambda_0}+1)}f(e_1)\\ \vdots \\ -\lim_{\lambda\rightarrow\lambda_0}\lambda^{-(d_{\lambda_0}+1)}f(e_n)\end{array}\right]=\left[\begin{array}{c}
-f_{\lambda_0}(e_1)\\ \vdots \\ -f_{\lambda_0}(e_n)\end{array}\right].
\end{equation}

On the other hand, by Definition \ref{Def:Bilimit}, the condition of statement~\ref{Th.Item2} implies that $f_{\lambda_0}(e_i)$ is homogeneous with respect to the standard dilation for $\lambda_0=0$ with degree $d_0<0$ and for $\lambda_0=\infty$ with degree $d_\infty>0$. Thus, by \eqref{eq:aux1}, $F_{\lambda_0}(e)$ is homogeneous with respect to the standard dilation for $\lambda_0=0$ with degree $d_0<0$ and for $\lambda_0=\infty$ with degree $d_\infty>0$, which implies that $F(e)$ is homogeneous in the $0-$limit with degree $d_0<0$ and in the $+\infty-$limit with degree $d_\infty>0$, in accordance to Definition \ref{Def:Bilimit}. Moreover, according to~Theorem~\ref{Prop.AsympStability}, if the origin is a globally asymptotically stable equilibrium of $\dot{x}=-f_{\lambda_0}(x)$  then the network's closed-loop system $\dot{x}=-F_{\lambda_0}(\mathcal{Q}(\mathcal{X}_k) x)$ converges to a globally asymptotically stable equilibrium. Then, it follows from Theorem~\ref{Th:Fixed} that the equilibrium of the closed-loop system is globally fixed-time stable.
\end{proof}

The following corollary states, based on Theorem \ref{Th.FixedFinite}, that protocol~\eqref{Eq:ConsFiniteMine} is finite-time convergent and protocol~\eqref{Eq:ConsFixedMine} is fixed-time convergent. These are particular protocols of the analyzed class, but more finite-time and fixed-time protocols can be derived.

\begin{corollary}
\label{Th:FiniteFixed}
Let $\sigma(t)=r$, $\forall t\geq t_0$, and let $\mathcal{X}_r$ be a connected graph and let $g_1(x)=kx$, $g_2(x)=k\sign(x)$, $g_3(x,\alpha)=k\lfloor x\rceil^\alpha$ and $g_4(x,p,q)=k_1\lfloor x\rceil^p+k_2\lfloor x\rceil^q$. 
\begin{enumerate}
\item If a consensus protocol $u_i$ is obtained from $f(x)=l_1g_1(x)+l_2g_3(x,\alpha)$ where $\alpha\in(0,1)$, $l_1\geq0$ and $l_2>0$, following direction~\eqref{Eq.NonlinearConsensusB}, i.e. $u_i=f(e_i)$, then $u_i$ is a continuous consensus algorithm with finite-time convergence. \label{itm:1}
\item If a consensus protocol $u_i$ is obtained from $f(x)=l_1g_1(x)+l_2g_2(x)+l_3g_3(x,\alpha)$ where $\alpha\in(0,1)$, $l_1,l_3\geq0$ and $k_2>0$, following direction~\eqref{Eq.NonlinearConsensusB}, i.e. $u_i=f(e_i)$, then $u_i$ is a discontinuous consensus algorithm with finite-time convergence. \label{itm:1a}
\item If a consensus protocol $u_i$ is obtained from $f(x)=l_1g_1(x)+l_2g_3(x,\alpha)+l_3g_4(x,p,q)$, $l_1,l_2\geq0$, $l_3>0$, $q>\alpha>p$, $q>1>p>0$, following direction~\eqref{Eq.NonlinearConsensusB}, i.e. $u_i=f(e_i)$, then $u_i$ is a continuous consensus algorithm with fixed-time convergence. \label{itm:2}
\item If a consensus protocol $u_i$ is obtained from $f(x)=l_1g_1(x)+l_2g_2(x)+l_3g_3(x,\alpha)+l_4g_4(x,p,q)$ where $l_1,l_3\geq0$, $l_2,l_4>0$, $q>\alpha>p$ and $q>1$, following direction~\eqref{Eq.NonlinearConsensusB}, i.e. $u_i=f(e_i)$, then $u_i$ is a discontinuous consensus algorithm with fixed-time convergence. \label{itm:2a}
\end{enumerate}
\end{corollary}
\begin{proof}

Notice that, with respect to the standard dilation, $g_1(x)$ is homogeneous of degree $d_1=0$, $g_2(x)$ is homogeneous of degree $d_2=-1$, $g_3(x,\alpha)$ is homogeneous of degree $\alpha-1$. Thus, for statement~\ref{itm:1}, $f(x)$  can be written as the sum of two homogeneous functions where the smallest degree is $\alpha-1<0$. Thus, the proof for statement~\ref{itm:1} follows from Theorem~\ref{Th.FixedFinite}. The same argument applies for statement~\ref{itm:1a}, but since $k_2>0$ the smallest degree is $d_2=-1$ from $g_2(x)$. 

To prove the statement \ref{itm:2}, it is easy to verify that $f_0(x)$ defined as
\begin{equation}
\label{Eq:f0Fixed}
f_{0}(x)=\lim_{\lambda\rightarrow0}\lambda^{-(d_{0}+1)}f(\lambda x)=l_4k_1\left\lfloor  e_i \right\rceil^p
\end{equation}
is homogeneous of degree $p-1<0$ with respect to the standard dilation. Thus, $f(\bullet)$ is homogeneous in the $0-$limit with degree $d_0=p-1<0$. In a similar way
\begin{equation}
\label{Eq:finftyFixed}
f_{\infty}(x)=\lim_{\lambda\rightarrow+\infty}\lambda^{-(d_{\infty}+1)}f(\lambda x)=l_4k_2\left\lfloor  e_i \right\rceil^q
\end{equation}
is homogeneous of degree $q-1>1$ with respect to the standard dilation. Thus, the vector field~\eqref{ConsDynamic} is homogeneous in the $+\infty-$limit with degree $d_\infty=q-1$. Furthermore, the origin of $\dot{x}=-l_4k_1\left\lfloor x \right\rceil^p$ and $\dot{x}=-l_4k_2\left\lfloor x \right\rceil^q$ is a globally asymptotically stable equilibrium. 
 
Thus, according to Theorem~\ref{Th.FixedFinite}, statement \ref{itm:2} holds. Statement \ref{itm:2a} follows from a similar argument by noticing that if $k_2>0$ then $f(x)$ is homogeneous in the $+\infty-$limit with degree $d_\infty=q-1>1$ and homogeneous in the $0-$limit with degree $d_\infty=-1$.
\end{proof}

\begin{remark}
The use of homogeneity theory for finite-time and fixed-time convergence analysis does not provide a bound for the convergence-time. However, this approach will allow to demonstrate that protocols derived by following direction~\eqref{Eq.NonlinearConsensusB} (for instance derived from $f(x)$ in Table~\ref{Tab:Comparison}) achieve finite/fixed-time convergence even under dynamic networks.
\end{remark}

\subsection{Consensus over dynamic networks switching among connected topologies}

In the proof of Theorem \ref{Prop.AsympStability} it was shown that the function \eqref{Eq.LyapunoFunc} is a Lyapunov function, valid for any given connected topology. On the other hand, the stability theory for switching systems \citep{Liberzon2003} states that a switching system, composed of a collection of nonlinear systems and an arbitrary switching signal determining the currently evolving nonlinear system, is asymptotically stable if there exists a Lyapunov function valid for all the nonlinear systems in the collection. In this way, \eqref{Eq.LyapunoFunc} is a common Lyapunov function for a dynamic network under arbitrary switching, provided the communication topology is always connected, and thus it can be proved that the consensus state is a globally asymptotic equilibrium of the dynamic network. This is formally stated in the following theorem.

\begin{theorem}
\label{TheoremSwConsensus}
Consider a dynamic network $\mathcal{X}_{\sigma(t)}=\langle\mathcal{F},\sigma\rangle$ such that $\sigma(t)\in\{1,...,m\}$ and $\forall r\in\{1,...,m\}$ the graph $\mathcal{X}_r$ is connected.
Consider a consensus algorithm defined by direction \eqref{Eq.NonlinearConsensusB} and a continuous function $f:\mathbb{R}\rightarrow\mathbb{R}$ such that the origin is a globally asymptotically stable equilibrium of the system $\dot{x}=-f(x)$.
Then, the consensus state is a globally asymptotically stable equilibrium of the network's closed-loop system under an arbitrary switching signal $\sigma(t)$.
Moreover, the consensus algorithm
\begin{enumerate}
\item is a finite-time consensus algorithm if the vector field $f(x)$ can be written as $f(x)=f_1(x)+\ldots+f_k(x)$ and for each $i=1, \ldots, k$, the vector field $f_i(x)$ is homogeneous of degree $d_i$ with respect to the standard dilation and $d_1 <\cdots < d_k$ with $d_1<0$. \label{Th.Item1}
\item is a finite-time consensus algorithm if $f(x)$ is a piecewise function such that there exist a nonzero constant $b$ where $\forall x\in\{x:|x|<b\}$, $f(x)$ satisfies condition~(\ref{Th.Item1}). \label{Th.Item3} 
\item is a fixed-time consensus algorithm if the vector field $f(x)$ is homogeneous in the $0-$limit with degree $d_0<0$, homogeneous in the $+\infty-$limit with degree $d_\infty>0$ and the origin is a globally asymptotically stable equilibrium of the dynamic systems $\dot{x}=-f_0(x)$ and $\dot{x}=-f_\infty(x)$ (where $f_0$ and $f_{\infty}$ are obtained from~\eqref{Eq:Hom} with $\lambda_0=0$ and $\lambda_0=+\infty$, respectively). \label{Th.Item2}
\end{enumerate}
\end{theorem}
\begin{proof}

According to Theorem~\ref{Prop.AsympStability}, the Lyapunov function~\eqref{Eq.LyapunoFunc} asymptotically converges to zero regardless of the current connected topology $\mathcal{X}_r$, i.e. $V(x)$ defined as in~\eqref{Eq.LyapunoFunc} is a common Lyapunov function. Thus, by~\cite[Theorem 2.1]{Liberzon2003}, the equilibrium of the network's closed-loop system is globally asymptotically stable under arbitrary switching of the communication topology. Moreover, since the graph $\mathcal{X}_\sigma$ is connected, $\mathcal{Q}(\mathcal{X}_{\sigma{(t)}})x=0$ implies that $x\in \ker \mathcal{Q}(\mathcal{X}_\sigma{(t)})$, i.e. $x_1=\ldots=x_n$ and consensus is achieved.

The proof for finite/fixed-time stability follows the same argument as in Theorem~\ref{Th.FixedFinite}, i.e., by using homogeneity and Theorem~\ref{Th.FTStability} and Theorem~\ref{Th:Fixed} for finite-time convergence and fixed-time convergence, respectively.
\end{proof}

\begin{remark}
Notice that, for the case of switching among connected graphs, the convergence of the consensus algorithm \eqref{Eq:ConsSignMine}, obtained from \eqref{Eq.SignFunction} following direction~\eqref{Eq.NonlinearConsensusB}, is independent of the network topology, because if $x_j=x_{max}$ and $x_k=x_{min}$ then $\dot{V}=-2k$ with $V(x)$ as in~\eqref{Eq.LyapunoFunc}, regardless of the network topology or the number of nodes. However, this steady convergence rate is not obtained for the consensus algorithm~\eqref{Eq:ConsSign} because a neighbor $x_i$ of $x_j=x_{max}$ (resp. $x_k=x_{min}$) may satisfy $\sign(e_i)=\sign(e_j)$. Thus, $\dot{V}$ will have different values that depend on the topology and the state of the neighbors.
\end{remark}

\begin{example}
\label{Example1}
Consider a network composed of 10 vertices and two different graphs, $\mathcal{X}_0$ and $\mathcal{X}_1$. Let $\mathcal{X}_1$ be such that the $i$-th vertex is adjacent to the \mbox{$j=(i+1)(mod\ 10)$} vertex, where $x(mode\ 10)$ stands for the common residue of $x$ modulo $10$, and  let $\mathcal{X}_0$ be such that the $i$-th vertex is adjacent to the \mbox{$j=(i+3) (mod\ 10)$} vertex. Let $\sigma(t)$ be the switching signal and let the initial condition be $x(t_0)=\left[\begin{array}{cccccccccc}
0 & -5 & 10 & 3 & -8 & -2 & 5 & 3 &-1 & 4
\end{array}
\right]$. Figure \ref{FigEx1} shows the convergence of the finite-time consensus algorithm~\eqref{Eq:ConsFiniteMine} in Table~\ref{Tab:DiferentConsensus}, obtained from~\eqref{Eq.FiniteTimeFunction} following direction~\eqref{Eq.NonlinearConsensusB}, under the graph topology $\mathcal{X}_{\sigma(t)}$ and switching signal $\sigma(t)$.

\begin{figure}
\centering
\includegraphics[width=13cm,trim={0 0 0 3.3cm},clip]{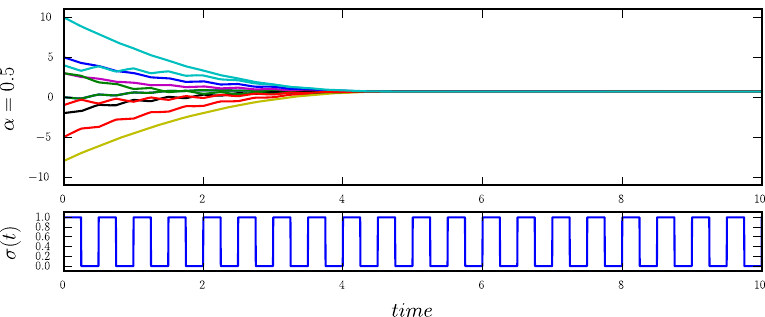}
\caption{Convergence of the consensus algorithm for Example~\ref{Example1} with $k=1$.}
\label{FigEx1}
\end{figure}

\end{example}

\subsection{Consensus over dynamic networks switching among disconnected topologies}

Theorem~\ref{TheoremSwConsensus} guarantees consensus along the network under arbitrary switching. A particular case occurs when $\sigma(t)=i$, $\forall t\in[0,\infty)$, i.e. the system remains in the same topology without switching. Thus, a necessary condition for consensus under arbitrary switching signal is that each possible topology is connected. Otherwise, each connected component could reach a different consensus since there will not be communication among components. This connectivity condition for each network topology can be relaxed by requiring a connected graph in a ``joint sense". This is formalized in the following.

\begin{definition}
Let $\mathcal{X}_{\sigma(t)}=\langle\mathcal{F},\sigma\rangle$ be a dynamic network with $\sigma(t)\in\{1,\ldots,m\}$. The switching signal $\sigma(t)$ is said to generate a $\tau$-jointly connected graph if there exists $\tau<\infty$ such that for all $\bar{t}\geq 0$, the graph $\bar{\mathcal{X}}$ with vertex set
$\mathcal{V}(\bar{\mathcal{X}})=\mathcal{V}(\mathcal{X}_{\sigma(t)})$ and edge set
$\mathcal{E}(\bar{\mathcal{X}})=\mathcal{E}(\mathcal{X}_{\sigma(t_i)})\cup \cdots \cup \mathcal{E}(\mathcal{X}_{\sigma(t_k)})$ is connected, where $t_i,\ldots,t_k$ are the successive switching times in the time interval $[\bar{t},\bar{t}+\tau]$. 
\end{definition}

\begin{theorem}
\label{TheoremJoint}
Let $\mathcal{X}_{\sigma(t)}=\langle\mathcal{F},\sigma\rangle$ be a dynamic network such that the switching signal $\sigma(t)$ generates a $\tau$-jointly connected graph. 

Consider a consensus algorithm defined by direction \eqref{Eq.NonlinearConsensusB} and a continuous function $f:\mathbb{R}\rightarrow\mathbb{R}$ such that the origin is a globally asymptotically stable equilibrium of the system $\dot{x}=-f(x)$.
Then, the consensus state is a globally asymptotically stable equilibrium of the consensus evolution~\eqref{ConsDynamic}.

\end{theorem}
\begin{proof}
Similarly as in the proof of Theorem~\ref{Prop.AsympStability}, we will show the convergence of $x$ to a consensus state, under the switched dynamic topology $\mathcal{X}_{\sigma(t)}$, by using the candidate Lyapunov function~\eqref{Eq.LyapunoFunc} and showing that $g(t)=V(x(t))$ along the trajectory of the system converges to zero provided that the switching signal generates a $\tau$-jointly connected graph. To this end, notice that if $\mathcal{X}_k$ is the current graph topology, not necessarily connected, then according to Lemma~\ref{lem:Sergej}, $V(x)$ in~\eqref{Eq.LyapunoFunc} is continuously differentiable except on a set of points $\{t_{1}^*, t_2^*,...\}$.

Thus, the time derivative of $V(x)$ along the trajectory of~\eqref{ConsDynamic} in the time interval $(t_i^*,t_{i+1}^*)$ is given by
\begin{equation}
\dot{V}=-(|f(e_j)|+|f(e_k)|)\leq 0 \ \ \ \forall t\in (t_i^*,t_{i+1}^*).
\label{Eq.TheoSw}
\end{equation}

It was shown in the proof of Theorem~\ref{Prop.AsympStability}, that, if the current graph topology is connected, the equality in~\eqref{Eq.TheoSw} holds for a nonzero interval only if consensus is achieved, i.e. $x_1=\ldots=x_n$. However, if the current graph topology $\mathcal{X}_k$ is not connected then the equality can hold, for a nonzero interval, whenever $\exists x_j,x_k$ and connected components $\mathcal{K}$ and $\mathcal{L}$ of $\mathcal{X}_k$ such that $x_j=\max(x_1,\ldots,x_n)$, $x_k=\min(x_1,\ldots,x_n)$, $j\in\mathcal{K}$, $k\in\mathcal{L}$ and consensus is achieved along $\mathcal{K}$ and $\mathcal{L}$.

Nonetheless, since $\sigma(t)$ generates a $\tau$-jointly connected graph within any time interval of length $\tau$, a graph $\mathcal{X}_{\hat{k}}$ will become active when there exists a node $\hat{j}$ adjacent to a node $\hat{i}$ such that $x_{\hat{j}}=\max\{x_1,\ldots,x_n\}$ and $x_{\hat{j}}>x_{\hat{i}}$ (a similar argument applies for a node $x_{\hat{k}}=\min\{x_1,\ldots,x_n\}$). Thus, for each $x\notin\ker \mathcal{Q}(\bar{\mathcal{X}})$ such that $\dot{V}=0$ and every time interval $[t,t+\tau]$ of length $\tau$ there exists a graph $\mathcal{X}_{\hat{k}}$, that will become active in $[t,t+\tau]$ such that $\dot{V}(x)<0$. Thus, by LaSalle's invariance principle~\cite{Khalil2002} it follows that $g(t)=V(x(t))$ will asymptotically converge to zero for every solution $x(t)$ of~\eqref{ConsDynamic}.
\end{proof}

\begin{example}
\label{Example2}
Consider a dynamic network composed of 10 vertices and 10 graphs. Let $\mathcal{X}_i$, $i\in\{1,\ldots,10\}$, be a graph with vertex set $\mathcal{V}(\mathcal{X}_i)=\{1,\ldots,10\}$ and edge set $\mathcal{E}(\mathcal{X}_i)=\{ij,ji\}$ such that $j=i+1(mod\ 10)$. The initial condition is $x(t_0)=\left[\begin{array}{cccccccccc}
0 & 5 & 3& 2&4&-9&10&5&-5&-3
\end{array}
\right]$. The evolution of the consensus algorithm \eqref{Eq:ConsFiniteMine} on the switched dynamic network $\mathcal{X}_{\sigma(t)}$ for two different switching signals $\sigma_1(t)=\lfloor t\rfloor (mod\ 10)+1$ and $\sigma_2(t)=\lfloor 100t\rfloor (mod\ 10)+1$ (where $\lfloor \bullet\rfloor$ denotes the floor function) is shown in Figure~\ref{FigEx2}-a) and Figure~\ref{FigEx2}-b), respectively. Notice that the switching signals as defined above generate $\tau$-jointly connected graphs with $\tau=10$ and thus consensus is achieved. Moreover, notice that $\sigma_2$ has a faster switching frequency than $\sigma_1$, thus the behavior of the network with $\sigma_2$ seems to be smoother.

\begin{figure}[t]
\centering
\includegraphics[width=13cm]{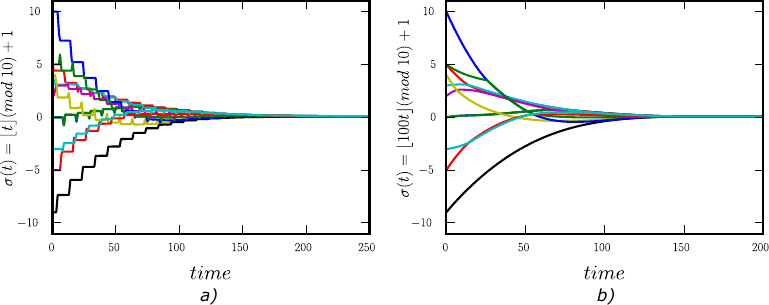}
\caption{Convergence of the consensus algorithm for Example~\ref{Example2} and switching signal generating a $\tau$-jointly connected graph.}
\label{FigEx2}
\end{figure}
\end{example}

\begin{corollary}
\label{CorollaryJoint}
Let $\mathcal{X}_{\sigma(t)}=\langle\mathcal{F},\sigma\rangle$ be a dynamic network, and $\tau$ a finite number such that within each time interval of length $\tau$, a strongly connected graph is active during a nonzero interval.
Consider a consensus algorithm defined by direction \eqref{Eq.NonlinearConsensusB} and a continuous function $f:\mathbb{R}\rightarrow\mathbb{R}$ such that the origin is a globally finite-time (respectively, fixed-time) stable equilibrium of the system $\dot{x}=-f(x)$.
Then, the consensus state is a globally finite-time (respectively, fixed-time) stable equilibrium of the consensus evolution~\eqref{ConsDynamic}.
\end{corollary}

\begin{remark}
In this paper, the analysis has been focused on the class of protocols that follow direction~(\ref{Eq.NonlinearConsensusB}), however, a similar analysis can be performed for the class of protocols that follow direction~(\ref{Eq.NonlinearConsensusA}), by using the same candidate Lyapunov function~\eqref{Eq.LyapunoFunc}.
\end{remark}

\section{Benchmark: Convergence time vs Graph connectivity}
\label{Sec.IllustrativeExample}

In this section, experiments are performed to evaluate the convergence time of a network's closed-loop system under different consensus algorithms. In particular, it is investigated how the convergence time increases when the graph's algebraic connectivity decreases. 

\subsection{Description and Motivation}

The motivation of this work is to analyze a class of algorithms that may work in a wide range of (possible unanticipated) situations. Imagine for instance a company developing low-power nodes of a sensor network, which must achieve consensus to provide an output sensing value, and whose interest is in enabling users to apply its solution in either small or large networks with minimum additional configurations. For a given topology, and assuming that a bound is known for the initial consensus error (a realistic assumption in sensor consensus), the gains of any consensus protocol can be adjusted to obtain a proper convergence (or settling) time. However, an important desired property of the implemented consensus algorithm is that the convergence time is maintained within an acceptable range, without the need of additional configuration, when the network's connectivity changes by either the connection or disconnection of sensors. This property is investigated in this benchmark, by comparing the convergence time of different protocols when the algebraic connectivity changes.

In detail, we compare algorithms based on the direction~\eqref{Eq.NonlinearConsensusB} against existing finite-time and fixed-time consensus algorithms for dynamics networks that were designed following direction~\eqref{Eq.NonlinearConsensusA}. Generally, the convergence time of a consensus algorithm grows when the algebraic connectivity of the graph decreases, which occurs when the network size increases. However, it will be shown that such increment in the convergence time is slower in nonlinear algorithms based on the direction~\eqref{Eq.NonlinearConsensusB} than in algorithms based on the direction~\eqref{Eq.NonlinearConsensusA}. Thus, the analyzed direction \eqref{Eq.NonlinearConsensusB} can be applied, with the same parameters selection, to graphs with either high or low algebraic connectivity, still achieving consensus in a satisfactory amount of time.

\subsection{Methodology}

The next methodology was used to benchmark the direction \eqref{Eq.NonlinearConsensusB} in two experiments that illustrate how the convergence time of each algorithm increases as the algebraic connectivity decreases. To this aim, switched networks are generated in such a way that the algebraic connectivity decreases as the number of nodes increases. For this, circular undirected graphs of $n$ nodes are defined, which are denoted by $\mathcal{C}_n$, satisfying $\lambda_2(\mathcal{C}_n)=2-2\cos\left(2\pi/n\right)$ (where $\lambda_2(\bullet)$ denotes the second eigenvalue of the Laplacian of the argument network).

\begin{itemize}
\item In the first experiment, the finite-time consensus algorithms of Table~\ref{Tab:DiferentConsensus} are compared. Namely, the algorithm~\eqref{Eq.OtherFTConsesusAlg}, proposed in~\cite{Wang2010}, versus the algorithm~\eqref{Eq:ConsFiniteMine}, which applies the same nonlinear function of \eqref{Eq.OtherFTConsesusAlg} but following direction~\eqref{Eq.NonlinearConsensusB}.

\item In the second experiment, the fixed-time consensus algorithms of Table~\ref{Tab:DiferentConsensus} are compared. Namely, the algorithm \eqref{Eq.FixedConsesusAlg}, proposed in~\cite{Zuo2014}, versus the algorithm~\eqref{Eq:ConsFixedMine}, which applies the same nonlinear function of \eqref{Eq.FixedConsesusAlg} but following direction~\eqref{Eq.NonlinearConsensusB}.

\item A dynamic network is considered, described by two undirected graphs of $n$ nodes, $\mathcal{X}_0$ and $\mathcal{X}_1$, where $\mathcal{X}_0$ is such that $(i,j)\in\mathcal{E}(\mathcal{X}_0)$ if and only if $j-i\equiv \pm 1 (mod\ n)$ and $\mathcal{X}_1$ is such that $(i,j)\in\mathcal{E}(\mathcal{X}_0)$ if and only if $j-i\equiv \pm h (mod\ n)$, where $h=max(\{h\in\{1,\ldots,\lfloor n/2\rfloor\}| n\ (mod\ n/2)\equiv 1\})$. The switching signal is given by $\sigma(t)=\lfloor 5t\rfloor(mod\ 2)$. An example of these graphs for the case of a graph with $n=25$ nodes is illustrated in Figure~\ref{Fig:ExampleGraph}.

\begin{figure}
\centering
\def\svgwidth{10cm}
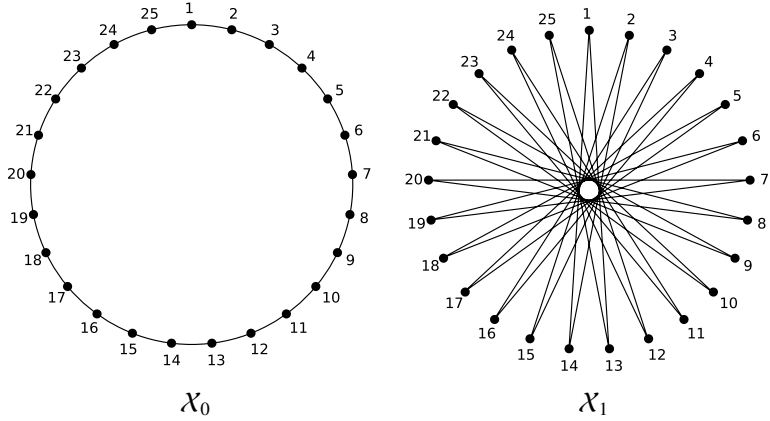
\caption{Example of the undirected switching graph  $\mathcal{X}_{\sigma(t)}$, $\sigma(t)\in{0,1}$ of 25 nodes used for benchmarking.}
\label{Fig:ExampleGraph}
\end{figure}

\item The initial conditions are set equally for the different algorithms using the linear congruential generator~\citep{Brunner1999},
$$
z_{i+1}=rz_i+s\ \ (\text{mod}\ M) \ \ n, \ \
x_i(t_0)=l\dfrac{z_i}{M}-m, \ \ \ i=1,\ldots,n
$$
such that $z_0=M$ and $r=45$, $s=1$, $M=1024$, $l=20$ and $m=10$ and $n$ is the number of nodes in the graph. This iterative procedure produces a pseudo--random sequence of initial conditions $x_i(t_0)$ in the interval $[-10,10]$.

\item The exponents of the consensus protocols are set equal, $\alpha=0.5$ for~ the first experiment and $q=\frac{3}{2}$, $p=\frac{1}{2}$ for the second experiment. Additionally, the gains are experimentally set (for the second experiment $k=k_1=k_2$) such that in a network of $25$ nodes, both algorithms achieve $V(x)=0.05$ at $1.00s$.

\item To measure the control effort of each approach, the Integrated Squared Control Effort (\textit{ISCE}) of the network is computed as $$E_{tot}(t)=\sum_{i=1}^n{E_i}(t)\text{, where }E_i(t)=\left(\int_{t_0}^{t}u_i^2\right)^{\frac{1}{2}}.$$

\item Experiments are performed varying from 25 to 1000 nodes. The convergence time and the \textit{ISCE} of each test are compared.

\item The simulations are performed in OpenModelica\textsuperscript{\textregistered} using Euler's integration method with interval $0.0001s$.
\end{itemize}

\subsection{Results}

The results for the first experiment, comparing the finite-time consensus algorithms in Table~\ref{Tab:DiferentConsensus}, is presented in Figure~\ref{ComparisonPlotsF}~a). It is important to highlight that, even if both algorithms achieve $V=0.05$ at time $t_f=1$ for $n=25$, the \textit{ISCE} of~\eqref{Eq.OtherFTConsesusAlg} is $E_{tot}(t_f)=361.31$ whereas the \textit{ISCE} of the proposed method~\eqref{Eq:ConsFiniteMine} is $E_{tot}(t_f)=273.57$.

The results for the second experiment, comparing the fixed-time consensus algorithms in Table~\ref{Tab:DiferentConsensus}, are presented in Figure~\ref{ComparisonPlotsF}~b). The \textit{ISCE} of \eqref{Eq.FixedConsesusAlg} in a network of 25 nodes is $E_{tot}(t_f)=616.15$, while the \textit{ISCE} of \eqref{Eq:ConsFixedMine} for the same network is $E_{tot}(t_f)=588.15$.

The results of these experiments suggest, first that the \textit{ISCE} required to achieve consensus at a given time is lower by following the direction \eqref{Eq.NonlinearConsensusB}; second, that the convergence time growing with the decreasing of the algebraic connectivity is significantly slower with the algorithms based on the direction \eqref{Eq.NonlinearConsensusB} than with the finite/fixed time algorithms of~\cite{Wang2010,Zuo2014} based on the direction \eqref{Eq.NonlinearConsensusA}. As notice in Table~\ref{Tab:Comparison}, previous results on finite-time and fixed-time consensus algorithms obtained following direction~\eqref{Eq.NonlinearConsensusB} does not justify the convergence to the consensus state in this example, since those results are restricted to static networks.

\begin{figure}
\centering
\def\svgwidth{12cm}
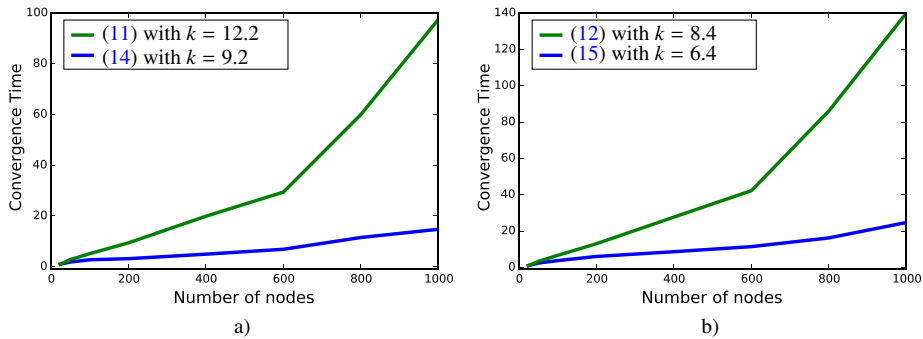
\caption{Benchmark comparing the convergence as the number of nodes in a network increases. a) First experiment: Finite-time consensus algorithms \eqref{Eq.OtherFTConsesusAlg}  vs  \eqref{Eq:ConsFiniteMine}. b) Second experiment: Fixed-time consensus algorithms~\eqref{Eq.FixedConsesusAlg} vs \eqref{Eq:ConsFixedMine}. Algorithms \eqref{Eq.OtherFTConsesusAlg} and \eqref{Eq.FixedConsesusAlg} follow direction \eqref{Eq.NonlinearConsensusA}, whereas \eqref{Eq:ConsFiniteMine} and \eqref{Eq:ConsFixedMine} follow direction \eqref{Eq.NonlinearConsensusB}.}
\label{ComparisonPlotsF}
\end{figure}

\section{Conclusions and Future Work}
\label{Sec.Conclusions}
n this work, a class of consensus algorithms for dynamic networks with finite/fixed-time convergence were analyzed by using homogeneity theory and switching stability theory. In particular, it was shown that the analyzed class, identified as direction \eqref{Eq.NonlinearConsensusB}, in which a nonlinear function of the consensus error is evaluated per each node, achieves finite/fixed-time consensus even if the communication topologies are disconnected. This feature is an essential advantage concerning other finite-time consensus algorithms that require that the sum of the time intervals for which the topology is connected be sufficiently large. Thus, the analyzed class allows the application of finite/fixed-time consensus algorithms with intermittent connections.

Among the advantages of the analyzed consensus algorithms over other previously proposed finite/fixed-time consensus algorithms for dynamic networks, the analyzed algorithms are computationally simpler, use lower control effort to achieve consensus at a given time and have slower growth in the convergence time as the algebraic connectivity decreases.

Future work concerns the analysis of the considered consensus class under noisy measurements as well as the implementation of its discrete version over robotic swarms. Moreover, the extension to high-order agents will be studied.


\end{document}